\documentclass[twocolumn,
 reprint, 
 amsmath,amssymb,
]{revtex4}
\usepackage{graphicx}
\usepackage{dcolumn}
\usepackage{amsthm}
\usepackage{hyperref}
\usepackage{url}
\usepackage{physics}
\usepackage{hyperref}
\usepackage{cleveref}
\newtheorem{definition}{Definition}[section] 
\newtheorem{theorem}{Theorem}[section] \newtheorem{lemma}{Lemma}[section]

\usepackage{enumitem}
\usepackage{xcolor}
\usepackage{latexsym} 

\usepackage[mathlines]{lineno}

\begin{document}


\title{Kernel Method based on Non-Linear Coherent State}

\author{Prayag Tiwari $^{1}$ }

\email{prayag.tiwari@dei.unipd.it}

\author{Shahram Dehdashti $^{2}$}
\email{shahram.dehdashti@qut.edu.au}

\author{Abdul Karim Obeid $^{2}$}
\author{Massimo Melucci $^{1}$}
\author{Peter Bruza $^{2}$}
\email{p.bruza@qut.edu.au}
\affiliation{$^{1}$ Department of Information Engineering, University of Padova, Padova, 35131, Italy.}
\affiliation{ $^{2}$School of Information Systems, Queensland University of Technology, Brisbane 4000, Australia.}

\thanks{The first two authors contributed equally in this research.}





\date{\today}

\begin{abstract}
In this paper, by mapping datasets to a set of non-linear coherent states, the process of encoding inputs in  quantum states as a  non-linear feature map is re-interpreted.  
As a result of this fact that  the  Radial Basis Function is recovered when data is mapped to a complex Hilbert state represented by coherent states,  
non-linear coherent states can be considered as natural generalisation of  associated kernels. By considering the non-linear coherent states of a quantum oscillator with variable mass, we propose a kernel function based on generalized hypergeometric functions, as  orthogonal polynomial functions. The suggested kernel  is implemented with support vector machine on two well known datasets (make\_circles, and make\_moons) and outperforms the baselines, even in the presence of high noise. In addition, we study impact of geometrical properties of feature space, obtaining by non-linear coherent states, on the SVM classification task, by using considering the Fubini-Study metric of associated coherent states. 
\end{abstract}

                              
\maketitle


\section{Introduction}\label{sec:level1}

Quantum machine learning is a rapidly growing field of investigation. It can be argued that developments are being driven from two directions. Firstly, quantum computers offer the promise of massive improvement in the speed of computational processing \cite{aaronson2013bosonsampling,tillmann2013experimental,brod2019photonic}. Secondly, the mathematical framework of quantum mechanics is increasingly been seen as a suitable framework for designing algorithms that aren't constrained by Boolean algebra and logic \cite{pitowsky1994george,vourdas2019probabilistic}.The reasons that support the latter claim are many, e.g., the linearity of the Schr\"{o}dinger equation \cite{schrodinger1987schrodinger,tsutsumi1987schrodinger}, which leads to the definition of superposed states in complex Hilbert spaces with an ‘interference’ term affecting probabilities. 
Consider modelling the dependence that measurement outcomes have on the preparation of states: duly named ‘contextual scenarios’ \cite{kochen1975problem}. There is also the novelty of correlations observed through entanglement, discord, {\it etc} in quantum mechanics that go beyond classically correlated structures \cite{zurek2000einselection,girolami2014quantum}, as well as quasi-distributions, which occur in phase space, so-called Wigner distributions \cite{simon1987gaussian,lorce2011quark}, and achieve negativity - this is not possible in Kolmogorovian probability theory. All of the preceding introduce the potential for access to an information space greater than that of classical alternatives \cite{goh2018geometry,pusey2019contextuality,dehdashti2020irrationality,uprety2020quantum}. This is encouraging for scientists wishing to apply quantum formalism within machine learning (ML) as a generalisation of probability theory.\\
\indent In ML, kernel methods \cite{shawe2004kernel,zelenko2003kernel,soentpiet1999advances,hofmann2008kernel,evgeniou2005learning,campbell2002kernel} are a class of categorization algorithms. Used within a wide range of methods and algorithms, they include  the support vector machine (SVM) \cite{amari1999improving,wang2005support,noble2006support}, kernel operators with principal components analysis (PCA) \cite{scholkopf1997kernel}, spectral clustering \cite{dhillon2004kernel}, canonical correlation analysis \cite{akaho2006kernel}, linear adaptive filters \cite{liu2009information}, and ridge regression \cite{an2007face}. Indeed, kernel methods are proving to assist also in deep neural networks, for which there are many recently published works \cite{cho2009kernel,belkin2018understand}. There exist vast prospects of kernel methods in ML due to the non-linear nature of the underlying data. Within application settings, the data are usually non-separable, for which the requirement of the kernel then becomes transformations (of the data) into higher dimensions where it may be (linearly) separable as can be seen in Figure \ref{fig:kernelpic}.
 One of the well-known
kernel functions in ML, the Radial Basis function (RBF),
is defined by $K(x,x^{\prime}) = \exp\left[-|x-x^{\prime}|^{2}/2\sigma^{2}\right]$
where $x$ and $x^{\prime}$ are two sample elements, and $\sigma$  controls the decision boundary \cite{musavi1992training,buhmann2000radial,orr1996introduction}. It is worthwhile to mention that the RBF  can be understood as an inner product of the linear  coherent state, see \cite{kubler2019quantum}, which is defined as an eigenstate of the  annihilation operator of a harmonic oscillator. 
The idea of using a quantum mechanics formalism in kernel methods was suggested by Schuld and Killoran, who introduced squeezed kernels in feature spaces \cite{schuld2019quantum}.  In fact, they defined  the feature space as a set of squeezed states so that  the kernel is obtained by inner products of squeezed states \cite{datko1970extending,gleason1957measures}.\\
\indent 
 In this paper, we express non-linear coherent states \cite{de1996nonlinear,man1997f,mancini1997even,roy2000new,sivakumar2000studies} as a quantum feature space, such that kernel functions are defined as their inner products. 
We show that the mathematical structure of non-linear coherent states provides infinite kernels. As an example of non-linear feature space, we investigate coherent states constructed by wave-functions of a quantum oscillator with variable mass. Generalized hypergeometric functions, as orthogonal polynomials, are identified as provide the associated kernel. Our proposed KMNCS has been demonstrated in an SVM classifier, along with the RBF and squeezed kernel as a baseline on two-well-known datasets (make\_moons, make\_circles). KMNCS is shown to outperform the baselines (squeezed and RBF kernels) even as we increase the noise in the dataset (which increases difficulty of generalisation). In addition, we study the geometrical properties of the feature space, by obtaining the Fubini-Study metric of non-linear coherent states. We show that the feature space of a non-linear coherent state of a oscillator with variable mass is a surface with negative curvature, which opens a new line of investigation of how the feature space's curvature affects the accuracy SVM classification. \\
\indent The rest of the paper is organised as follows: in Section \ref{sec2}, we briefly review the Kernel Method. Section \ref{sec3} introduces the previously mentioned coherent states. Section \ref{sec4} provides an overview on non-linear coherent states and introduces the coherent states of a quantum oscillator with variable mass. This allows the main contribution of the paper to be defined: a kernel function based on generalized hypergeometric functions. Section \ref{sec5} details the experimental design which allows the proposed kernel function to be empirically evaluated against two baseline: RBF and squeezed kernel. Section \ref{sec6} discusses the results of the empirical evaluation. Also, this section is examines the geometrical properties of a non-linear coherent state.  Finally, Section \ref{sec7} concludes the article.

\begin{figure}[t]
    \centering
    \includegraphics[width=8 cm, height=6 cm]{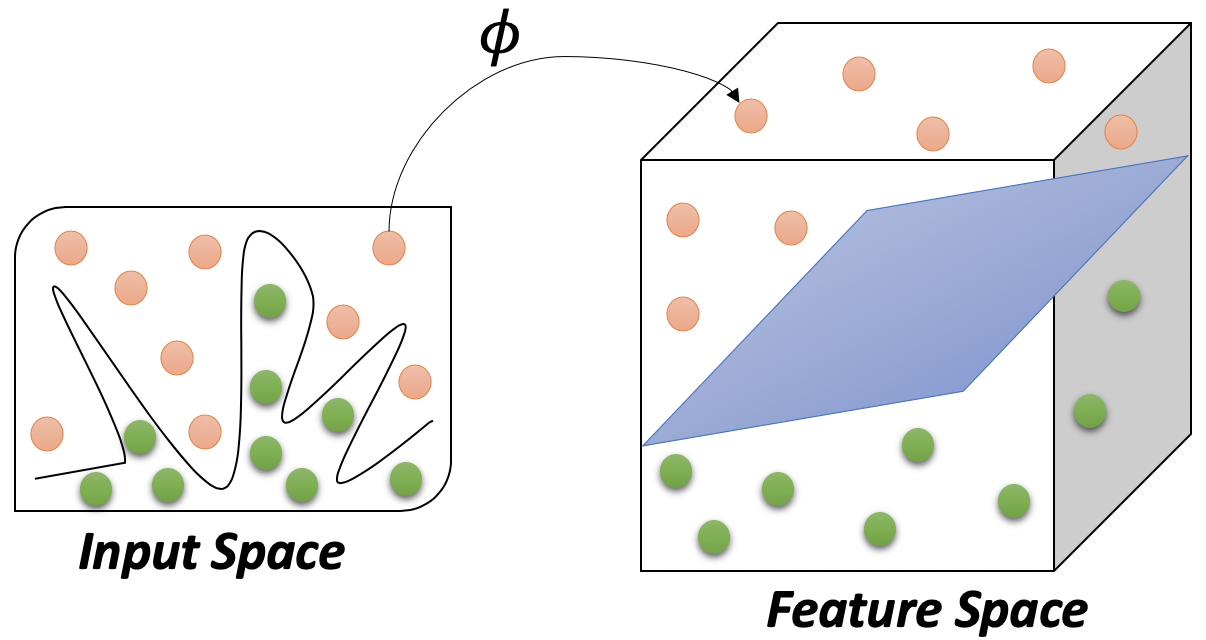}
    \caption{Kernel visualization from low dimension input space to high dimension feature space.}
    \label{fig:kernelpic}
\end{figure}

\section{Kernel Method}\label{sec2}
Traditional ML begins with a dataset of inputs $\mathcal{D} = \{x_{1},\cdots ,x_{M}\}$ drawn from a set $\mathcal{X}$.
The goal is to produce a predictive model that allows patterns to be discovered in yet to be observed data.
Kernel methods use the inner product $K(x, x^{\prime})$ between any given two inputs $x, x^{\prime} \in \mathcal{X}$, as a distance measure to build predictive models that assist with representing characteristics of a data distribution.  
\begin{definition}\label{de11}
 Let $\mathcal{X}$ be a nonempty set, 
called the index set,  and $\mathcal{H}$ by a Hilbert space of functions $\phi : \mathcal{X}\rightarrow \mathbb{R}$. Then $\mathcal{H}$ is called
a reproducing kernel Hilbert space endowed with the dot product, 
$ \bra{\cdot} \ket{\cdot} $,  if there exists a function $k : \mathcal{X}\times \mathcal{X \rightarrow \mathcal{R}}$ with the  two following properties: (i) $k$ has the reproducing property, i.e., $\bra{\phi}\ket{k(x,\cdot)}=\phi(x)$ for all $\phi \in \mathcal{H}$; (ii) $k$ spans $\mathcal{H}$ \cite{scholkopf2001learning}.

\end{definition}
 Note that in particular $\bra{k(x,\cdot)}\ket{k(x^{\prime},\cdot)}=k(x,x^{\prime})=k(x^{\prime},x)$ guarantees symmetry of the arguments of $k$. 
\begin{theorem}
\label{theorem1}
Let $\phi: \mathcal{X}\rightarrow \mathcal{F}$ be a feature map. Kernel function can be defined as the inner product of two inputs mapped to some feature space by, $\mathbf{k}(x, x^{\prime} ) =
 \langle\phi(x), \phi(x^{\prime}) \rangle$. 
\end{theorem}

\begin{proof}
The proof can be found in references  
\cite{hofmann2008kernel,kubler2019quantum}.
\end{proof}

\begin{theorem}
\label{theorem2}
Consider a feature mapping is $\phi: \mathcal{X}\rightarrow \mathcal{F}$ over some input set $\mathcal{X}$, which provide the basis to a complex kernel $\mathbf{k}(x, x^{\prime} ) =
 \langle\phi(x), \phi(x^{\prime}) \rangle$ that is defined on $\mathcal{F}$. The associated reproducing kernel Hilbert space can be written as $\mathcal{R}_{\mathbf{k}} = \{f:\mathcal{X}\rightarrow \mathbb{C}\}$   
\end{theorem}
\begin{equation}
    f(x) = \langle w, \phi(x^{\prime})\rangle,
\end{equation}
for all $w \in \mathcal{F}$ and $x \in \mathcal{X} $.
\begin{proof}
The proof can be found in  \cite{hofmann2008kernel} and \footnote{\url{http://www.gatsby.ucl.ac.uk/~gretton/coursefiles/lecture4_introToRKHS.pdf}}. 
\end{proof}

\section{Coherent State}\label{sec3}
A harmonic oscillator in quantum physics is described by the following Hamiltonian:
\begin{eqnarray}
H=\hbar \omega\left(\hat{n}+\frac{1}{2}\right)
\end{eqnarray}
in which $\hbar$ is the Planck's constant  and $\omega$ is the  angular frequency of the oscillator; the number operator  $\hat{n}$ is described by annihilation and creation operators, i.e., $\hat{n}=\hat{a}^{\dagger}\hat{a}$. The Sch\"orodinger equation gives the discrete energy  eigenvalue:
\begin{eqnarray}
H\ket{n}=E_{n}\ket{n}, \ n=0,1,2, \cdots  
\end{eqnarray}
in which eigenvalue $E_{n}=\hbar\omega(n+1/2)$ is associated by eigenstate $\ket{n}$. For simplicity, we consider $\hbar=1$ and $\omega=1$ in the rest of the paper. \\ 
\indent A coherent state is the specific quantum state of the quantum harmonic oscillator, often described as a state which has dynamics most closely resembling the oscillatory behavior of a classical harmonic oscillator, for example see \cite{ali2000coherent}. 
\begin{definition}
A coherent state is defined as superposition of number state  $\ket{n}$ as following:
\begin{eqnarray}\label{eq4}
\ket{\alpha}=e^{-|\alpha|^{2}/2}\sum_{n=0}^{\infty}\frac{\alpha^{n}}{\sqrt{n!}}\ket{n}
\end{eqnarray}
which number states satisfy  $\bra{n}\ket{m}=\delta_{n,m}$, where $\delta_{n,m}$ is Kronecker delta.
\end{definition}
Note that the inner product of two coherent states is given by:
\begin{eqnarray}\label{eq2}
\bra{\alpha}\ket{\beta}=e^{|\alpha-\beta|^{2}/2}
\end{eqnarray}
\begin{lemma} A harmonic oscillator coherent state adheres the following: 
\begin{enumerate}\label{lem1}
     \item It is obtained by operation of the displacement operator $D(\mathbf{\alpha})=\exp{\left(\alpha^{\ast}a^{\dagger}-\alpha a\right)}$ on a reference state: $|\alpha\rangle=D(\alpha)|0\rangle$, in which $\ket{\alpha}$ is given by the relation (\ref{eq4}).
  \item  It is an eigenvector of the annihilation operator, $a|\alpha\rangle=\alpha|\alpha\rangle$.
  \item It fulfills  the minimum uncertainty relation, i.e., $\Delta(\mathbf{x})=\Delta(\mathbf{p})=1/\sqrt{2}$, where $\Delta (\mathbf{x})=\sqrt{\langle \mathbf{x}^{2}\rangle-\langle \mathbf{x}\rangle^{2}}$, where $\langle \mathbf{x}\rangle=\bra{\alpha}\mathbf{x}\ket{\alpha}$. 
  \item It is over-complete, $K(\alpha,\alpha^{\prime})=\bra{\alpha}\ket{\alpha^{\prime}}\neq \delta(\alpha-\alpha^{\prime})$, i.e., the equation (\ref{eq2}), despite the fact that they fulfil the resolution of the identity, $\int d\mu(\mathbf{x}) |\mathbf{x}\rangle\langle \mathbf{x}|=\mathbb{I}$, which leads to the following relation:
  \begin{eqnarray}
 \int d\mu(\mathbf{x}) \langle\phi|\mathbf{x}\rangle\langle \mathbf{x}|\psi\rangle=\langle\phi|\psi\rangle. 
  \end{eqnarray}
\end{enumerate}
\end{lemma}
\begin{proof}
The proofs can be found in references  \cite{ali2000coherent,kubler2019quantum}.
\end{proof}

The latter, i.e., item  $4$ of Lemma \ref{lem1}, implies an arbitrary function can be expressible as a linear combination of kernel functions in a ``reproducing Hilbert space"  \cite{combescure2012coherent}. 
We should mention that the first above-mentioned property leads to define a displacement-type coherent states   \cite{ali2000coherent}, for  generalized annihilation and creation operators. Moreover  a Gazeau-Klauder coherent state is defined by the second property and fulfils the third property  \cite{ali2000coherent}. It should be noted that while the latter,  i.e., resolution of the identity, fulfils all types of coherent states. 

Recently Schuld and Killoran published a paper in which data is mapped into a feature space established by squeezed  states \cite{schuld2019quantum}. Squeezed states are  states that saturate the Heisenberg uncertainty principle; additionally, the quadrature variance of position and momentum depend on a parameter,  so-called squeezing parameter. The squeezing parameter causes the uncertainty to be squeezed  one of its quadrature components, while  stretched  uncertainty for the other component, i.e. $\Delta(\mathbf{x})=\exp{2\zeta}/\sqrt{2}$ and $\Delta(\mathbf{p})=\exp{-2\zeta}/\sqrt{2}$, where $\zeta$ is the  squeezing parameter. Therefore, the squeezing parameter controls uncertainty via a quadrature component, while the third and fourth properties of coherent states are preserved.
 However, as we mentioned,  one of the well-known kernel functions in ML, the Radial Basis function (RBF), is defined by  
\begin{equation}\label{eq3}
K(\mathbf{x},\mathbf{x}^{\prime};\sigma)=\exp{\left(-||\mathbf{x}-\mathbf{x}^{\prime}||^{2}/2\sigma^{2}\right)},    
\end{equation}
 where $\mathbf{x}$ and $\mathbf{x}^{\prime}$ are two sample elements, and $\sigma$ is understood as a free parameter. 
 However, drawing a comparison between relations (\ref{eq2}) and (\ref{eq3}), inspires  someone to interpret the RBF as an inner product of  two coherent states. This interpretation opens the door to define new kernels and consequently improve the kernel method \cite{kubler2019quantum}.

\section{Non-linear Coherent State}\label{sec4}
 As mentioned before,  some efforts have been devoted to study possible generalization of the quantum harmonic oscillator algebra \cite{man1997f}. A deformed algebra is a non-trivial generalization of a given algebra through the introduction of one or more deformation parameters, such that, in a certain limit of the parameters, the non-deformed algebra can be recovered. 
 
 A particular deformation of the W-H algebra led to the notion of $f$-deformed oscillator \cite{man1997f}. An $f$-deformed oscillator is a non-harmonic system where its dynamical variables (creation and annihilation operators) are constructed from a non-canonical transformation through 
      \begin{eqnarray}
        &&
        \hat{A}=\hat{a}f(\hat{n})=f(\hat{n}+1)\hat{a},\label{eqd2}\\
        &&
        \hat{A}^{\dag}=f^{\dag}(\hat{n})\hat{a}^{\dag}=\hat{a}^{\dag}f^{\dag}(\hat{n}+1).\label{eqd3}
      \end{eqnarray}
      where $f(\hat{n})$ is called deformation function by which non-linear properties of this system are governed. 
 An $f$-deformed oscillator is characterized by a Hamiltonian of the harmonic oscillator form,
 \begin{eqnarray}\label{3}
 \hat{H}=\frac{\omega}{2}( \hat{A} \hat{A}^{\dag}+ \hat{A}^{\dag} \hat{A}),
  \end{eqnarray}
    where  $ \hat{A}$ and $  \hat{A}^{\dag} $ are given in equations (\ref{eqd2})  and (\ref{eqd3}). In this Hamiltonian, $ \omega $ is frequency of harmonic oscillator and $ \hslash=m=1 $.  The deformed operators  satisfy the following commutation relation
    \begin{equation}\label{4}
      [\hat{A},\hat{A}^{\dag}]=(\hat{n}+1)f^{2}(\hat{n}+1)
        -\hat{n}f^{2}(\hat{n}).
    \end{equation}
    Relations (\ref{eqd2}) and (\ref{eqd3})  give eigenvalues of the Hamiltonian (\ref{3}) as follows: 
    \begin{equation}\label{5}
    E_{n}=\frac{\omega}{2}((n+1)f^{2}(n+1)
        +n f^{2}(n)).
    \end{equation}
    It is worth to mention that  by approaching deformation function into $1$, i.e., $f(n) \rightarrow 1$,  the non-deformation energy eigenvalues, $E_{n}=\omega(n+\frac{1}{2})$, and the non-deformed commutation relation $[\hat{a},\hat{a}^{\dag} ]=1$ are recovered. However, similar to the harmonic oscillator, it is possible to construct coherent states for the $f$-deformed oscillator. The non-linear transformation of the creation and annihilation operators leads naturally to the notion of non-linear coherent states or $f$-coherent states \cite{dehdashti2015decoherence,dehdashti2013decoherence,dehdashti2013coherent,dehdashti2015realization}.
\begin{definition}
Non-linear coherent states are defined as the right-hand eigenstates of the deformed annihilation operator $\hat{A}$ as follows \cite{ali2000coherent}:
\begin{equation}\label{9}
\hat{A}|\alpha\rangle_{f}=\alpha|\alpha\rangle_{f}.       
\end{equation}
From equation (\ref{9}) one can obtain an explicit form of the non-linear coherent states in the number state representation as,
\begin{equation}\label{10}
|\alpha\rangle_{f}=\frac{1}{\mathcal{N}} \sum_{n}^{M}  \frac{\alpha^{n}}{\sqrt{n! [f(n)]!}} |n \rangle,            
\end{equation}
 where $M$ can be finite, or infinite (corresponding to finite or infinite dimensional Hilbert space), $\alpha$ is a complex number and  $[f(n)]!=\prod_{i=0}^{n}f(i)$, with $[f(0)]!=1$; the normalization factor $\mathcal{N}$ is given by
     \begin{eqnarray}\label{11}
  \mathcal{N}=\left(\sum_{n=0}^{M} \frac{|\alpha|^{2n}}{n! [f(n)]!}\right)^{-1/2}. 
\end{eqnarray}
\end{definition}
  Therefore,  based on the definition \ref{de11}, we can define a kernel as the following: 
\begin{definition}
By  mapping   multi-dimensional input set  $\mathbf{x} = (x_{1},\cdots,x_{N})^{T} \in \mathbb{R}^{N}$ into  non-linear coherent states, defined by the relation (\ref{10}), so that they are  fulfilled with resolution of the identity,   a feature space is defined as
 \begin{equation*}
  \phi: (x_{1},\cdots, x_{N}) \rightarrow  |x_{1}\rangle_{f}\otimes |x_{2}\rangle_{f} \otimes \cdots \otimes |x_{N}\rangle_{f}.
 \end{equation*}
  In addition, the  associated kernel is obtained by the inner product as the following:
 \begin{eqnarray}\label{eq9}
 K(\mathbf{x},\mathbf{x^{\prime}})= \prod_{i=1}^{N}\  _{f}\langle x_{i} | 
 x_{i}^{\prime} \rangle_{f},
 \end{eqnarray}
 in which
 \begin{eqnarray}
_{f}\langle x_{i} | 
 x_{i}^{\prime} \rangle_{f}=\frac{1}{\mathcal{N}^{2}}\sum_{n=0}^{M}\frac{(x_{i}x_{i}^{\prime})^{m}}{m![f(m)]!}.
 \end{eqnarray}

\end{definition}
 \subsection{Non-Linear coherent state of an oscillator with variable mass}
 \begin{figure}[t]
    \centering
    \includegraphics[width=3.5 cm, height=4.5 cm]{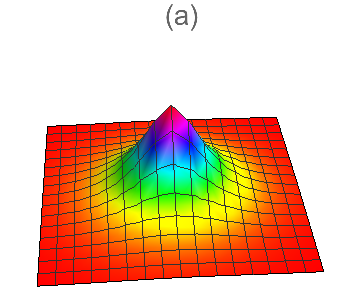}
    \includegraphics[width=3.5 cm, height=4.5 cm]{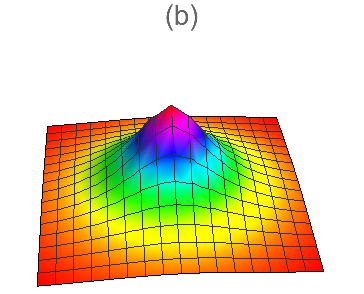}
    \caption{Schematic shape of the  kernel function (\ref{eq8}) for $k=-0.001$ and $k=-0.5$ in plots (a) and (b), respectively. The input $x$ is fixed at $(0, 0)$ and $x^{\prime}$ is varied.}
    \label{fig1}
\end{figure}
  The quantum version of a non-linear oscillator Hamiltonian  with variable mass is given by \cite{tchoffo2019supersymmetric}
 \begin{equation}
     H=\frac{1}{2\lambda}\left[-(1+\delta^{2}x^{2})\frac{d^{2}}{dx^{2}}-2\delta^{2}x\frac{d}{dx}+\frac{\lambda^{2}x^{2}}{1+\delta^{2}x^{2}}\right]
 \end{equation}
 in which where $\lambda$ is a real parameter and $\delta$ is constant that measures the force of the nonlinearity of the oscillator.
 By using  ladder operators, one can define a non-linear coherent state as the following:
 \begin{eqnarray}\label{eq1}
 \ket{x}=\frac{1}{\mathcal{N}(x)}
 \sum_{n=0}^{\infty}\frac{x^{n}}{\rho_{n}}\ket{n}
 \end{eqnarray}
 in which $\rho_{n}=n! (-k)^{n} (2-1/k)_{n}$, which $(u)_{n}=u(u-1)\cdots (u+n-1)$ represents the Pochhammer symbol, $k=\delta^{2}/2\lambda$ and normalization factor $\mathcal{N}(x)$ is given by 
 \begin{eqnarray}
 \mathcal{N}^{2}(x)&=&\sum_{n=0}^{\infty}\left(\frac{(1/k)^{n}}{n!(2-1/k)_{n}}\right)^{2}|x|^{2n}\nonumber\\
 &=& _{0}F_{3}(;1,2-1/k,2-1/k;|x/k|^{2})
 \end{eqnarray}
 which $_{0}F_{3}(;1,2-1/k,2-1/k;|x/k|^{2})$ is a generalized hyper-geometric function.\\
 \indent By considering a multi-dimensional input set in a data set of vectors $\mathbf{x} = (x_{1},\cdots,x_{N})^{T} \in \mathbb{R}^{N}$, one can define the joint state of $N$ deformed coherent states,
 \begin{equation*}
  \phi: (x_{1},\cdots, x_{N}) \rightarrow  |x_{1},k\rangle\otimes |x_{2},k\rangle \otimes \cdots \otimes |x_{N},k\rangle. 
 \end{equation*}
 Therefore, the kernel  is defined as the following:
 \begin{eqnarray}\label{eq19}
 K(\mathbf{x},\mathbf{x^{\prime}})= \prod_{i=1}^{N}\langle x_{i};k | 
 x_{i}^{\prime};k \rangle,
 \end{eqnarray}
 in which
\begin{eqnarray}\label{eq8}
\langle x_{i};k | 
 x_{i}^{\prime};k \rangle= \frac{_{0}F_{3}(;1,2-1/k,2-1/k;x_{i}x_{i}^{\prime}/k^{2})}{ \mathcal{N}(x) \mathcal{N}(x^{\prime})}
\end{eqnarray}
Figure \ref{fig1} schematically illustrates the kernel function (\ref{eq8}), for $k=-0.001$ and $k=-0.05$.
\subsection{Geometrical properties of associated Hilbert Space constructed by Non-linear coherent state}
For understanding the rule of nonlinear parameter $k$, we will study the geometrical properties of above-mentioned feature spaces.
We can define the line element of the feature space by
using the Fubini-Study metric \cite{bengtsson2017geometry}.
\begin{definition} A suitable metric between two Hilbert space vectors, e.g, $\ket{\psi}$ and $\ket{\phi}$, is defined by  
\begin{eqnarray}
d(\ket{\psi},\ket{\phi})=\min || \ket{\psi}-e^{i\alpha}\ket{\phi}||, \ 0\leq \alpha \leq 2\pi.
\end{eqnarray}
The infinitesimal form of this metric is given by the Fubini-Study metric:
\begin{eqnarray}\label{eq25}
ds^{2}= ||d\ket{x}||^{2}-||\bra{x}d\ket{x}||^{2}.
\end{eqnarray}
\end{definition}
The following gives the Fubini-Study metric of a non-Linear coherent state (\ref{eq1}).
\begin{theorem}
 The Fubini-Study metric of a non-linear coherent state (\ref{eq1}) is a surface with non-zero curvature with the following metric:
\begin{eqnarray}\label{eq26}
ds^{2}=\Omega(r)\left(dr^{2}+r^{2}d\phi^{2}\right),
\end{eqnarray}
in which 
\begin{eqnarray}
\Omega(r)=\frac{1}{2}\left[\frac{\partial^{2}_{r}\mathcal{N}(r)}{\mathcal{N}(r)}+\frac{\partial_{r}\mathcal{N}(r)}{r\mathcal{N}(r)}-
\left(\frac{\partial_{r}\mathcal{N}(r)}{\mathcal{N}(r)}\right)^{2}\right].
\end{eqnarray}
with $\partial_{r}\equiv \frac{\partial}{\partial r}$.
\end{theorem}
\begin{proof} We consider the  $x=re^{i\phi}$,  using the definition (\ref{eq25}), directly leads to the metric (\ref{eq26}).\\
By using the definition of  Christoffel symbol,
\begin{eqnarray}
\Gamma_{ab}^{c}=\frac{1}{2}g^{cd}\left(\partial_{a}g_{bd}+\partial_{b}g_{ad}-\partial_{d}g_{ab}
\right), \ a,b,c= r,\phi
\end{eqnarray}
with the Einstein summation rule, the non-zero  Christoffel symbols are given by 
\begin{eqnarray}
\Gamma_{rr}^{r}&=&\frac{1}{2}
\partial_{r} \ln \Omega(r), \\
\Gamma_{\phi \phi}^{r}&=&-\left(r+\frac{r^{2}}{2}
\partial_{r} \ln \Omega(r)\right), \\
\Gamma_{\phi r}^{\phi}&=&\frac{1}{r}+\frac{1}{2}
\partial_{r} \ln \Omega(r).
\end{eqnarray}
Hence, the non-zero Ricci tensors are given by,
\begin{eqnarray}
R_{rr}&=&-\frac{1}{2}\left(\partial_{r}^{2}\ln \Omega(r)+\frac{1}{r}\partial_{r}\ln \Omega(r)\right)\\
R_{\phi \phi}&=&-\frac{r}{2}\left(\partial_{r}\ln \Omega(r)+r\partial_{r}^{2}\ln\Omega(r)\right)
\end{eqnarray}
Hence, the Ricci scalar, $R=g^{ab}R_{ab}$, is obtained as
\begin{eqnarray}\label{eq34}
R=- \Omega(r)^{-1}\left(\partial_{r}^{2}\ln \Omega(r)+\frac{1}{r}\partial_{r}\ln \Omega(r)\right)
\end{eqnarray}
\end{proof}
 \begin{figure}[t]
    \centering
    \includegraphics[width=8 cm]{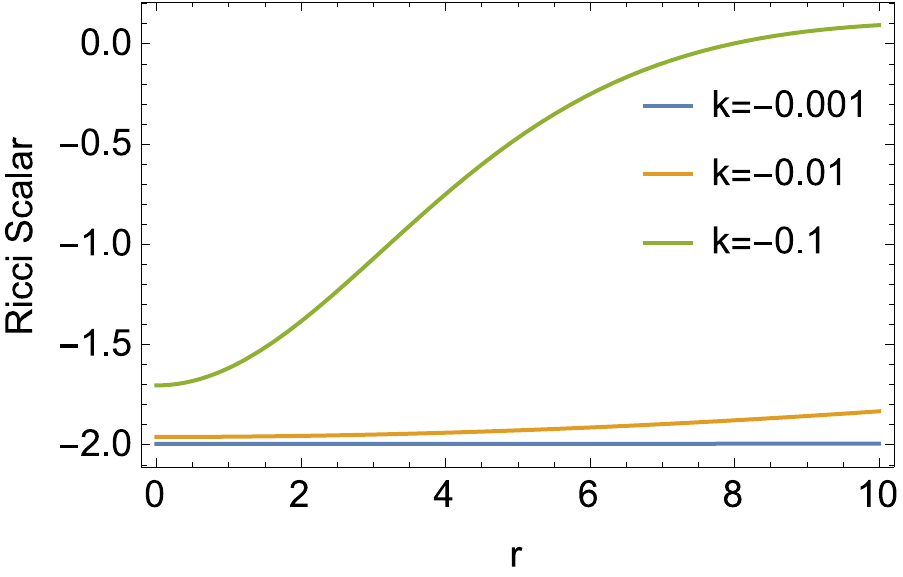}
    \caption{Ricci scalar of the Fubini-Study metric of a feature space constructed by a oscillator with variable mass   as function of $r$, for different values of $k$.}
    \label{fig3}
\end{figure}
\indent It is worthwhile to note that the metric (\ref{eq26}), which describes the feature space, is a surface that is conformal with flat space, that is   conformally preserves angles, while lengths can be changed. In fact, this metric illustrates a two-dimensional curved space, depending on conformal function $\Omega(r)$. By considering the non-linear coherent state (\ref{eq1}) and normalized coefficient (\ref{eq2}), Fig. \ref{fig3} represents the Ricci scalar for different values of $k$. This figure indicates that  the feature space is a surface with negative curvature. Decreasing the value $k$ causes the Ricci curvature to increase.  In other words, increasing value of $k$ causes the feature space and the associated kernel to approach a flat space and RBF kernel, respectively.    Despite the fact that the feature space constructed on a non-linear coherent state is a non-zero curved space, RBF, which is constructed from a linear coherent state, is a kernel on  flat space, with the zero Ricci scalar. 

\section{Experimental Design}\label{sec5}
 In the following, we empirically evaluate the  KMNCS, i.e. the relation (\ref{eq19}),  RBF kernel, i.e., the relation (\ref{eq3}), 
and the Squeezed kernel, given by
\begin{eqnarray}\label{eq35}
K_{Sq}=\prod_{i=1}\left[\frac{\sech c \sech c^{\prime} }{1-e^{i(x-x^{\prime})}\tanh c \tanh c^{\prime}}\right]^{1/2}.
\end{eqnarray}

\indent SVM uses a kernel function to define a decision boundary for separating the data points. Generally, the hyper parameter $C$ in the Gaussian kernel is used to optimize the performance of SVM, as a cost function connected with mis-classifications of data points in feature space of the training set. We kept the hyperparameter $C=1$ as the optimal value.\\ 
\indent
In order to provide a comprehensive picture of the performance of the kernels, i.e.,  kernels (\ref{eq3}), (\ref{eq35}) and (\ref{eq19}), noise was systematically applied to the input data. Applying noise implies adding it to the target variable. For example, if the noise parameter is say $0.2$, a standard deviation of $0.2$ would be observed (originating from Gaussian noise) in the output. When data points become inseparable due to noise, it becomes more challenging for the classifier to accurately classify the data points.\\
\indent 
We have used synthetic datasets (simple toy datasets) which are commonly used to check the performance of kernels. The two datasets ($make\_moons$\footnote{\url{https://scikit-learn.org/stable/modules/generated/sklearn.datasets.make_moons.html#sklearn.datasets.make_moons}} and $make\_circles$\footnote{\url{https://scikit-learn.org/stable/modules/generated/sklearn.datasets.make_circles.html#sklearn.datasets.make_circles}}) are taken from sklearn. There is some flexibility in regard to each dataset. Random noise can be introduced by adjusting different parameters for each dataset: the \textbf{Moons} dataset generates two half circles with the noise parameters affecting 'interleave', the \textbf{Circles} dataset generates concentric circles also affected by 'interleave' via the noise parameter. The decision region for class 1 is color-coded 'red', and 'blue' for class 0. We have also recorded differing values of $flip\_y$ which is an inbuilt parameter in $make\_classification$ \footnote{\url{https://scikit-learn.org/stable/modules/generated/sklearn.datasets.make_classification.html#sklearn.datasets.make_classification}} ( Generate a random n-class classification problem) also from sklearn. A large $flip\_y$ supplements the noise effects, making accurate classification even more challenging. The $factor$ and $random\_state$ in the setting have also been recorded. 
We have divided the data into  $60 \%$ training set and $40 \%$ test set. 
We have also used 5 fold cross validation during training in order to avoid overfitting problems. We have examined KMNCS using different values of $k$ in order to evaluate the classification performance and to understand how the decision boundaries are formed. 

\begin{figure}[t]
    \centering
    \includegraphics[width=9 cm, height=12 cm]{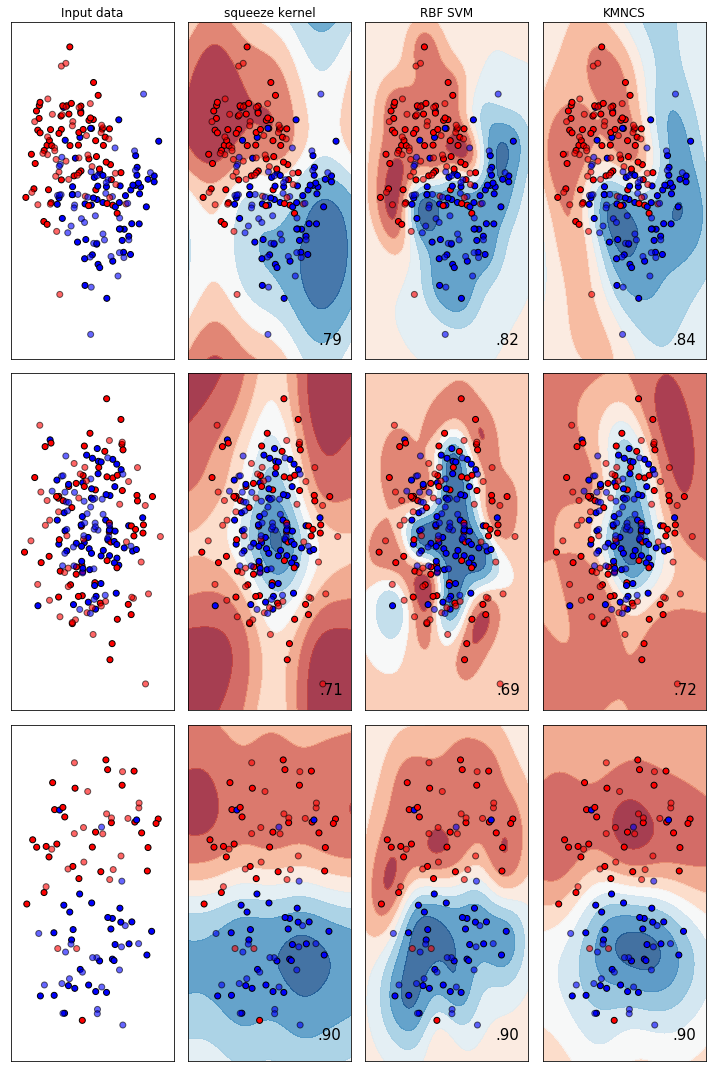}
    \caption{Kernel visualization (Squeezed, RBF kernels and KMNCS) for $k=-0.001$, while the noise is $0.4$ on both the make\_moons and make\_circles datasets.}
    \label{fig:kernel_visualization}
\end{figure}

\begin{figure}[t]
    \centering
    \includegraphics[width=9 cm, height=12 cm]{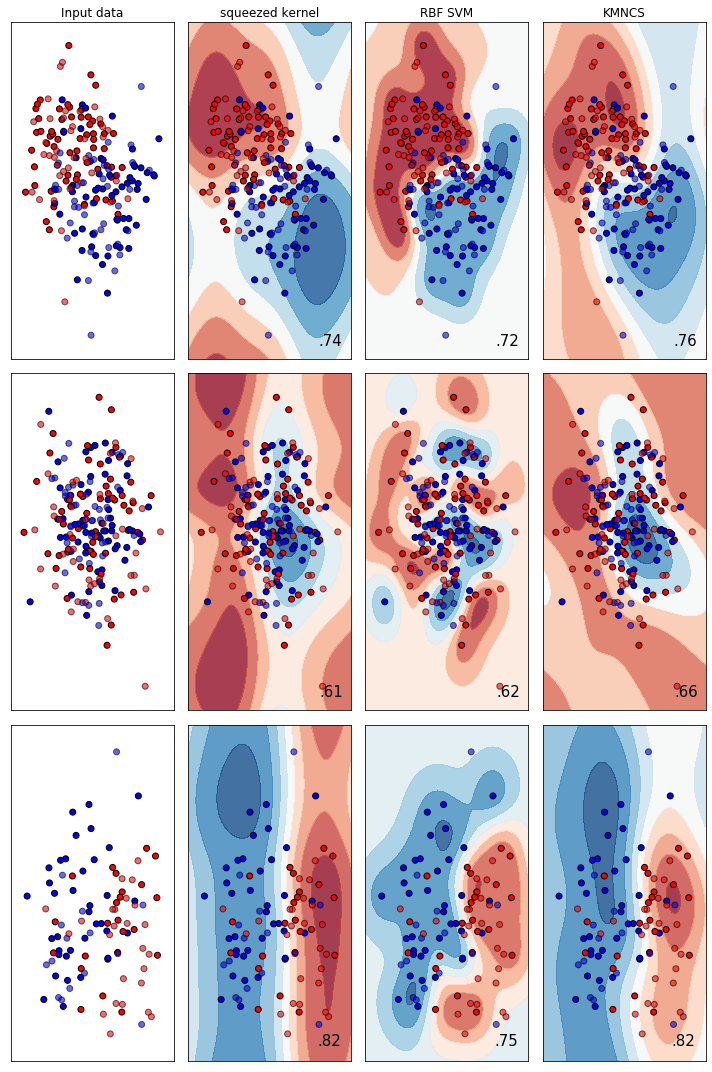}
    \caption{Kernel visualization (Squeezed, RBF, and KMNCS) for $k=-0.1$, while  noise is $0.5$ on the make\_moons dataset, and  $0.7$ on the make\_circles dataset.}
    \label{fig:kernel_visualization1}
\end{figure}

\begin{figure}[t]
    \centering
    \includegraphics[width=9 cm, height=12 cm]{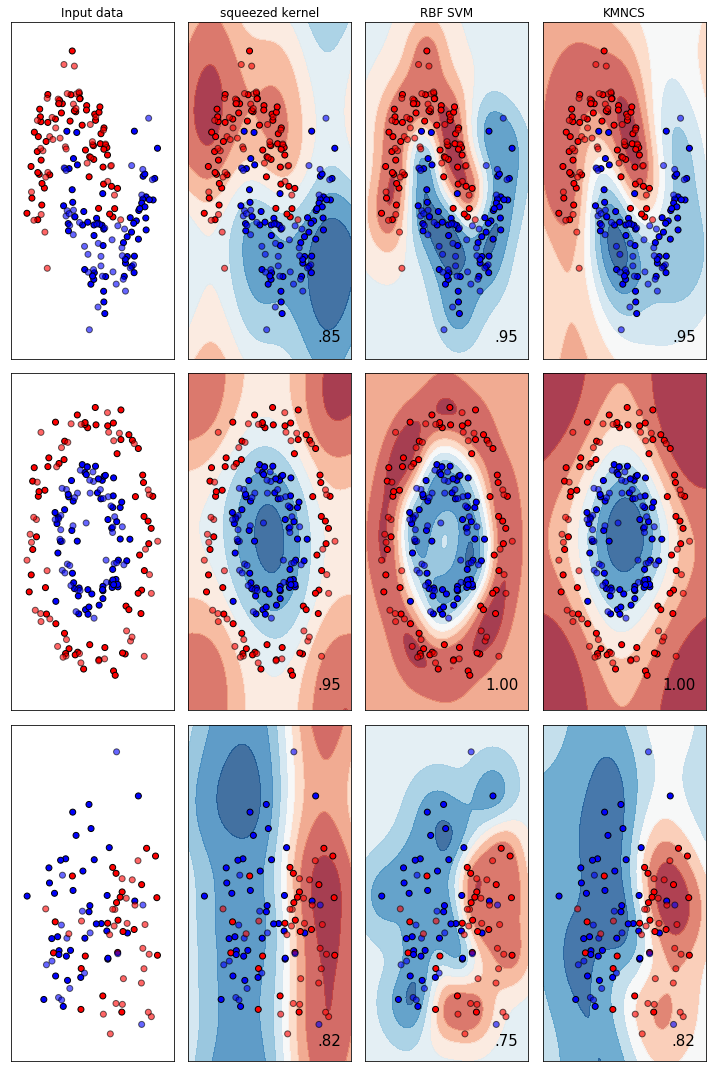}
    \caption{Kernel visualization (Squeezed, RBF, and KMNCS) for $k=-0.01$, while noise is $0.2$ on the make\_moons dataset, and $0.1$ on make\_circles dataset.}
    \label{fig:kernel_visualization2}
\end{figure}

\section{Results}
\label{sec6}
The sklearn package includes a ``fit" method that is used for informing the model by applying a training set. To compute the score by cross-validation of SVM, ``cross\_val\_score" is used also from sklearn with a ``5"-fold cross-validation. \\
\indent As we can see from Table \ref{tab:accuracy}, KMNCS has been tested on several values of $k$. If we check the value of $k=-0.001$ with very low noise, i.e. $=0.1$, all the three kernels provide almost the same classification accuracy as can be seen in Figure \ref{fig:kernel_visualization2} (Here the input data is separable and it is an easy task for the classifier to classify them). Interestingly, KMNCS has better performance than both baselines for both datasets when $k=-0.001$ and noise is $0.4$. This performance in the presence of such noise demonstrates effectiveness of KMNCS. The decision boundary is depicted in Figure \ref{fig:kernel_visualization}. 

When we increase the noise effect in the data, then classification becomes more challenging i.e., when noise is $0.5$ then KMNCS outperforms the squeezed and RBF kernel as can be seen in Figure \ref{fig:kernel_visualization1}. Note that it is possible to increase the accuracy of the classifier if we select $k=-0.01, -0.1$ with the same level of noise, that is $0.5$, as it can be seen in the Table \ref{tab:accuracy}. The KMNCS performance is deemed better than the baselines as data is inseparable (due to $50\%$ noise in the target variable). Again, where noise is  $0.7$ and $k=-0.1$, then KMNCS provides better performance than the two baseline classifiers due to clean decision boundaries where data is separable, as can be seen in Figure \ref{fig:kernel_visualization1}. 

We also tested on $3000$ samples as can be seen in Table \ref{tab:accuracy_with_3000samples}. KMNCS outperforms both baselines which suggests it is effective on larger samples by producing clean decision boundaries. \\
\indent
There is a trade-off with decision boundaries and accuracy. It is possible to achieve even higher accuracy if we continue to modify $k$, however this may obscure the decision boundaries ( despite high accuracy scores), which could be related to bias and variance issues. 
We also computed the score without the cross validation set with the results presented in Table \ref{tab:accuracy_without_crossvalidation}.  
\begin{table}[t]
\centering
\caption{Accuracy on Squeezed kernel, RBF, and \textbf{KMNCS} on the make\_moon and make\_circles datasets with the following parameters: n\_samples=200, random\_state=50, y\_flip=0.2, factor=0.5, n\_informative=2, 5-fold cross validation. Noise effects are introduced in the data to increase difficulty of accurate classification.}
\begin{tabular}{llll}
\hline
              & Parameters & moons      & circles \\ \hline
Squeezed Kernel & Noise=0.1   &0.90       &0.95   \\
              & Noise=0.2   &0.85       &0.86   \\
              & Noise=0.3   &0.82       &0.80   \\
              & Noise=0.4   &0.79       &0.71   \\
             & Noise=0.5    &0.74       &0.68     \\
             & Noise=0.7    &0.62       &0.61     \\
              \hline
RBF           & Noise=0.1   &0.99       &1.0    \\
              & Noise=0.2   &0.95      &0.90   \\
              & Noise=0.3   &0.89       &0.80   \\
              & Noise=0.4   &0.82       &0.69   \\
              & Noise=0.5   &0.72      &0.64  \\
              & Noise=0.7   &0.65      &0.62   \\
              \hline
\textbf{KMNCS}         
& Noise=0.1,  k=-0.001 &0.97  &\textbf{1.0} \\
              & Noise=0.2, k=-0.001  &\textbf{0.95}  &0.88 \\
              & Noise=0.2,  k=-0.01   &\textbf{0.95}  &0.89 \\
              & Noise=0.2,  k=-0.1    &0.94  &0.88 \\
              & Noise=0.3,  k=-0.001    &\textbf{0.90}  &0.79 \\
              & Noise=0.3,  k=-0.01    &\textbf{0.90}  &0.79 \\
              & Noise=0.4,  k=-0.001    &\textbf{0.84}  &\textbf{0.72} \\
              & Noise=0.4, k=-0.01    &\textbf{0.84}  &\textbf{0.72} \\
              & Noise=0.5, k=-0.001  &\textbf{0.75} &\textbf{0.68}\\
              & Noise=0.5, k=-0.01   &\textbf{0.76} &\textbf{0.69}\\
              & Noise=0.5, k=-0.1    &\textbf{0.76} &\textbf{0.70}\\
              & Noise=0.7, k=-0.001  &0.64 &\textbf{0.62} \\
              & Noise=0.7, k=-0.01   &0.62 &\textbf{0.62}\\
              & Noise=0.7, k=-0.1    &\textbf{0.65} &\textbf{0.66} \\
              \hline
\end{tabular}
\label{tab:accuracy}
\end{table}

\begin{table}[t]
\centering
\caption{Accuracy on same parameters from Table \ref{tab:accuracy} without cross validation}
\begin{tabular}{llll}
\hline
              & Parameters & moons      & circles \\ \hline
Squeezed Kernel & Noise=0.2  &0.90       &0.94   \\
             & Noise=0.5    &0.75       &0.70     \\
             & Noise=0.7    &0.68       &0.56     \\
              \hline
RBF           & Noise=0.2   &0.96       &0.90    \\
              & Noise=0.5   &0.76      &0.64  \\
              & Noise=0.7   &0.70      &0.54   \\
              \hline
\textbf{KMNCS} & Noise=0.2, k=-0.001 &0.99  &0.93 \\
              & Noise=0.5, k=-0.001  &0.75  &0.66 \\
              & Noise=0.7, k=-0.001   &0.71  &0.54 \\
              \hline
\end{tabular}
\label{tab:accuracy_without_crossvalidation}
\end{table}

\begin{table}[t]
\centering
\caption{Accuracy on 3000 samples with same parameters from Table \ref{tab:accuracy} as well as 5-fold cross validation}
\begin{tabular}{llll}
\hline
              & Parameters & moons      & circles \\ \hline
Squeezed Kernel & Noise=0.5  &0.81       &0.65   \\
             & Noise=0.8    &0.72       &0.55     \\
             & Noise=1.0    &0.68       &0.52     \\
              \hline
RBF           & Noise=0.5   &0.81       &0.64    \\
              & Noise=0.8   &0.73      &0.54  \\
             & Noise=1.0    &0.68       &0.51     \\
              \hline
\textbf{KMNCS} & Noise=0.5, k=-0.1 &\textbf{0.82}  &\textbf{0.65} \\
              & Noise=0.5, k=-0.01 &\textbf{0.81}  & \textbf{0.65}\\
              & Noise=0.8, k=-0.1 &\textbf{0.73}  &\textbf{0.56} \\
                & Noise=1.0, k=-0.1 &\textbf{0.69}  &\textbf{0.53} \\
              \hline
\end{tabular}
\label{tab:accuracy_with_3000samples}
\end{table}
\section{Discussion}
\label{sec7}
With respect to the general behaviour of  KMNCS, it can be said that the kernel is conducive to forming large decision boundaries, which is a major point of contrast when considering the RBF kernel,  as evident in Figure \ref{fig:kernel_visualization},\ref{fig:kernel_visualization1}. As previously mentioned, the hyper-parameter $C$ can be used to 
optimise an SVM classifier, as a cost function associated with mis-classification of elements in feature spaces of the ’training’ set. It  implies  the  maximisation  of  $C$  tightens  decision  boundaries  (so called ’hard’  margins),  and  was  introduced  by Boser {\it et al.} \cite{boser1992training}. In later work, 'hard' margins were found to fail on even slightly inseparable data. As a solution, \cite{schiilkop1995extracting} introduced a meaningful technique for the minimisation of $C$ that was found to enlarge the space covered by decision boundaries ('soft' margins). While this has largely benefited the field of SVM in contention with other classification techniques, the trade-off is that significantly 'soft' margins fail to classify data entirely - formally known as 
over-generalisation, large decision boundaries (as produced by the Squeezed kernel) become non-representative of the data. As a result, the objective is to minimise the hyper-parameter $C$, while maintaining the highest possible classification. Relating this back to the findings of this paper, the topological structure of the KMNCS is reminiscent of Ref. \cite{boser1992training}, surprisingly in cases of sparse data ( in this case the sparse $make\_moons$ dataset), without concern for $C$.
\section{Conclusions and Remarks}\label{seccr}
 In this paper,  we mapped  datasets into non-linear coherent states, as a non-linear feature space, constructed by a complex Hilbert space. We showed that the  RBF kernel is recovered when data is mapped to a complex Hilbert space represented by coherent states. Therefore, non-linear coherent states can be considered as natural generalized candidates for formalizing kernels. In addition, by considering the non-linear coherent states of a quantum oscillator with variable mass, we proposed a kernel function based on generalized hypergeometric functions. This idea suggests a method for obtaining a generalized kernel function which can be expressed by orthogonal polynomial functions on the one hand, and opens a new door for  using quantum formalism to specify quantum  algorithms in continuous variable quantum computing, on the other. In addition, we studied the geometrical properties of the surface in which the kernel lives.   We indicated that the feature space of a non-linear coherent state is a non-zero curved space,  despite the fact that the RBF kernel lives on feature space which is  flat. This method can potentially open a door for studying the impact of general curved space on the  machine learning methods more  generally, and the problem of classification more specifically.\\
\indent More generally, this research has demonstrated how quantum approaches to machine learning may prove beneficial. In practical usage, machine learning applications of quantum theory have begun involving developments of physical circuitry \cite{havlivcek2019supervised}. These contributions have begun to realise the quantum processing components required to build quantum computational devices designed solely for
feature classification. It is inspiring to think that eventually, SVM classification may ( with the assistance of quantum theory) be computed substantially faster than ever before.

\begin{acknowledgments}
This reserach has received funding from the European Union's Horizon 2020 research and innovation programme under the Marie Sklodowska-Curie grant agreement No 721321. 
It was additionally supported by the
Asian Office of Aerospace Research and Development (AOARD) grant: FA2386-17-1-4016.
\end{acknowledgments}

\nocite{*}

\bibliography{apssamp}

\end{document}